\documentclass{article}
\usepackage{amsmath,graphicx,mlspconf}
\usepackage{dsfont}
\usepackage{float}
\usepackage{xfrac}
\usepackage{color}
\usepackage{psfrag}
\usepackage{amsthm}
\usepackage{subfigure}
\usepackage[lined,boxed,commentsnumbered]{algorithm2e}
\usepackage{enumerate}
\usepackage{import}
\usepackage{epsfig}
\usepackage{amssymb}
\usepackage{url}

\copyrightnotice{U.S.\ Government work not protected by U.S.\ copyright}

\copyrightnotice{978-1-5090-0746-2/16/\$31.00 {\copyright}2016 Crown}

\copyrightnotice{978-1-5090-0746-2/16/\$31.00 {\copyright}2016 European Union}

\copyrightnotice{978-1-5090-0746-2/16/\$31.00 {\copyright}2016 IEEE}

\toappear{2016 IEEE International Workshop on Machine Learning for Signal Processing, Sept.\ 13--16, 2016, Salerno, Italy}

\DeclareMathAlphabet{\bm}{OML}{cmr}{bx}{it}
\DeclareMathAlphabet{\mathsf}{OT1}{cmss}{m}{n}
\DeclareMathAlphabet{\bs}{T1}{cmss}{bx}{sl}
\DeclareMathAlphabet{\ms}{T1}{cmss}{m}{sl}
\DeclareMathAlphabet{\mathpzc}{OML}{zplm}{m}{it}

\newcommand{\bg}[1]{\boldsymbol #1} %for bold greek variables

\newcommand{\bb}{\mathbb}
\newcommand{\rs}{\mathrm}
\newcommand{\mc}{\mathcal}
\newcommand{\bh}[1]{\hat{\bm #1}}
\newcommand{\ds}{\mathds}

\newtheorem{theorem}{Theorem}
\newtheorem{lemma}{Lemma}
\newtheorem{proposition}{Proposition}

\newtheorem{remark}{Remark}
\newtheorem{definition}{Definition}

\title{Towards optimal nonlinearities for sparse recovery using higher-order statistics}
\name{Steffen Limmer$^\star$ and S\l awomir Sta\'nczak$^{\star,\dagger}$
\thanks{Acknowledgments: This work was partially supported by the Deutsche Forschungsgemeinschaft (DFG) under Grant STA 864/8-1 and an AWS in Education Research Grant award. The authors would like to thank the anonymous reviewer \#3 for his/her valuable comments.}
}
\address{$^{\star}$ Network Information Theory Group,
    Technische Universit\"at Berlin.\\		
    $^\dagger$ Fraunhofer Institute for Telecommunications, Heinrich Hertz Institute,
	Berlin, Germany.}
\newif\ifcomments
\commentsfalse
\ifcomments
\newcommand{\com}[1]{{\color{red}{(\textbf{Comment: #1})}}}
\else
\newcommand{\com}[1]{}
\fi
\begin{document}
\ninept

\maketitle
\begin{abstract}
  We consider machine learning techniques to develop low-latency
  approximate solutions for a class of inverse problems. More
  precisely, we use a probabilistic approach to the problem of recovering sparse stochastic
  signals that are members of the $\ell_p$-balls. In this context, we
  analyze the Bayesian mean-square-error (MSE) for two types of
  estimators: (i) a linear estimator and (ii) a structured estimator
  composed of a linear operator followed by a Cartesian product of
  univariate nonlinear mappings. By construction, the complexity of
  the proposed nonlinear estimator is comparable to that of its linear
  counterpart since the nonlinear mapping can be implemented efficiently in hardware by means of look-up
  tables (LUTs). The proposed structure lends itself to neural
  networks and iterative shrinkage/thresholding-type algorithms
  restricted to a single iteration (e.g. due to imposed hardware or
  latency constraints). By resorting to an alternating minimization
  technique, we obtain a sequence of optimized linear
  operators and nonlinear mappings that converge in the MSE objective. The result
  is attractive for real-time applications where general
  iterative and convex optimization methods are infeasible.
\end{abstract}


\section{Introduction}
Precise error estimates and phase transitions play a crucial role in
the analysis of compressed sensing recovery algorithms, where the
objective is to recover an unknown $N$-dimensional real-valued vector
signal $\bm{x} \in\bb{R}^N$ from a measurement vector $\bm{y} \in \bb{R}^M$
given by \cite{AmLoMcTr13}\com{I think that we need a reference at the end of this
  sentence}
\begin{align}\label{equ:cs_meas}
{y}_m & =  \langle \bm{a}_m,\bm{x}\rangle, \quad \forall m \in
        \{1,\hdots,M\},\ M<N.
\end{align}
Here and hereafter,\footnote{We refer to the end of this section for
  some further notational conventions.}\com{I think that we need such
  a note either as a footnote or in text}
$\langle \bm\cdot,\cdot\rangle:\bb{R}^N\times\bb{R}^N\to\bb{R}$
denotes the inner product in the Euclidean space $\bb{R}^N$, while the
matrix $\bm{A}:= [\bm{a}_1,\hdots,\bm{a}_M]^T \in \bb{R}^{M \times N}$
is a dimensionality reducing linear map that may be given or designed depending on the particular application.\com{Note that I have added $M<N$ to the equation above
  since you write ``dimensionality reducing''} Motivated by the
seminal work \cite{AmLoMcTr13}, we study a probabilistic approach to
the above recovery problem,\com{I have changed this sentence because I
  do not understand what you mean by ``probabilistic approach to the
  geometry''.  My main problem is the word ``geometry'' in this
  context. I have the same problem in the abstract. Now it is in my
  opinion clear what is meant.} with the goal of assessing and
optimizing the expected performance for a certain class of nonlinear
estimators that can be implemented efficiently in hardware. In
contrast to \cite{AmLoMcTr13}, we assume that the measurement map
$\bm{A}$ is fixed and the randomness originates from a stochastic
model of the estimand $\bm{x}$.\com{I have removed
  ``$p_{\bs{x}}(\bm{x})$" because it is not needed here and the
  necessary notation/definitions have not been introduced until here}
It is therefore evident that the performance of any estimator
(resp. recovery algorithm) will be tightly coupled to the statistical
properties of the inner products in \eqref{equ:cs_meas} with the
sought sparse random vector $\bs{x}$.\com{This sentence came later and
  I have shifted it here} Among a myriad of models that have been
proposed to analyze sparse/compressible signals at different layers of
abstraction, the set of $k$-sparse signals
$\Sigma_k:=\{\bm{x} : \lVert \bm{x} \rVert_0 \leq k\}$, $k<N$,
is frequent choice in the field of approximation theory (see e.g. \cite{Cohen2009}).\com{This statement requires
  a reference} The set $\Sigma_k$ is however of Lebesgue measure zero
in $\bb{R}^N$,\com{I think that you need to add $\bb{R}^N$ because the
  Lebesgue measure depends on the embedment} which makes the
treatment within a unified probabilistic framework difficult.\com{Is this
  sentence okay?} To overcome this limitation, we study the recovery
of sparse stochastic signals from generalized unit balls $\mc{B}_\bm{p}$ that are
equipped with the desired sparsity inducing structure for
$\bm{p}<2\cdot \bold{1}$ and are closely related to the set $\Sigma_k$ \cite{Cohen2009} (see the definition of $\mc{B}_\bm{p}$ in Lemma \ref{lem:volbp} and Fig. \ref{fig:lpballs} for an illustration).\com{I
  think that we need a reference here} In this probabilistic $\mc{B}_\bm{p}$-model, the characteristic vector $\bm{p}$ adjusts the energy concentration in subsets of largest entries (in magnitude), i.e., the \emph{sparsity} of realizations $\bm{x}$. In practice, we may use techniques from parametric density estimation to obtain estimates of the sparsity level in terms of $\bm{p}$ given some dataset.
A review of selected existing and new results is provided in Sec. \ref{sec:lbpalls}.
% $\mc{B}_{p}:= \{\bm{x} : \sum_{n=1}^N \lvert x_n \rvert^{p} \leq 1\}$, with
%\begin{align}
%\lVert \bm{x} \rVert_p := 
%\begin{cases}
%\lim_{p \to 0} \sum_{n=1}^N \lvert x_n \rvert^p &\text{for } p=0, \\ %\big\lvert \rs{supp}(\bm{x}) \big\rvert &\text{for } p=0,\\
%\left( \sum_{n=1}^N \lvert x_n \rvert^p \right)^{\frac{1}{p}} &\text{for } p>0.
%\end{cases}
%\end{align}
%As such, the class $\Sigma_k$ can be seen as a limit case of out analysis. 
\begin{figure}
\centering \subfigure[$p=2$.]{
  \includegraphics[width=0.492\linewidth]{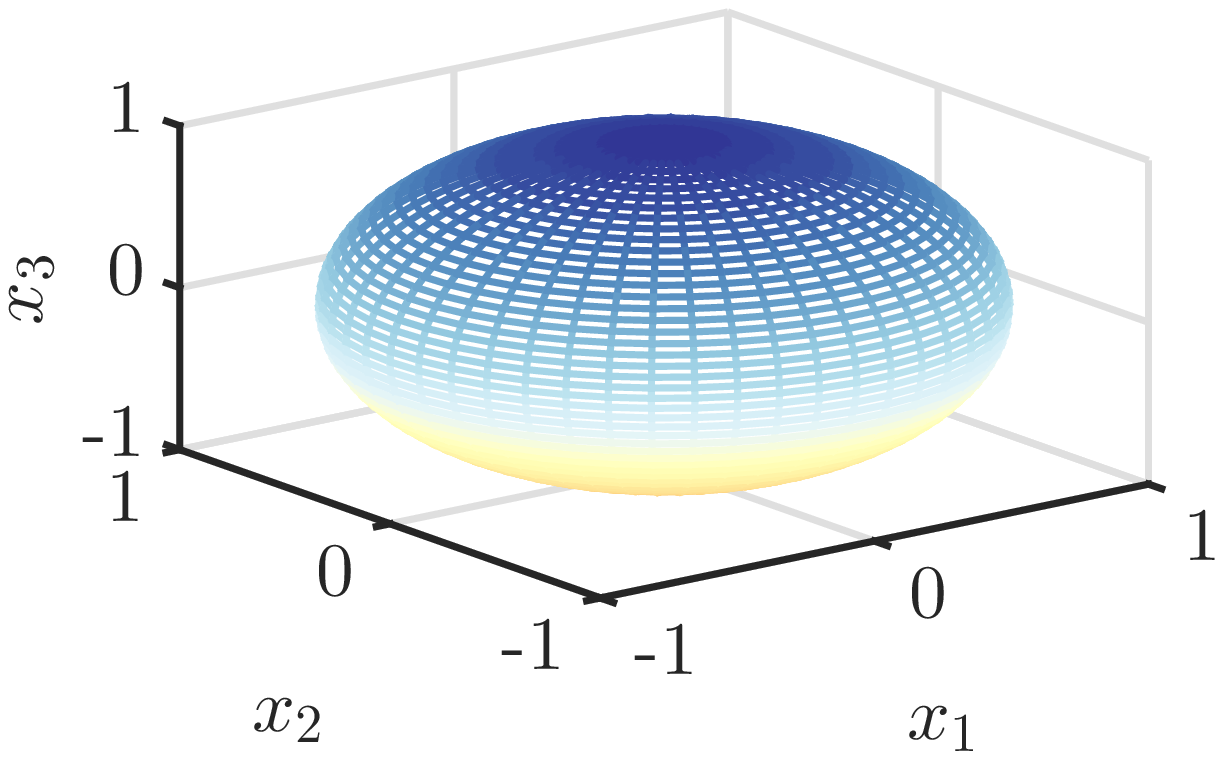}
} \hspace{-.55cm} \subfigure[$p=1$]{
  \includegraphics[width=0.492\linewidth]{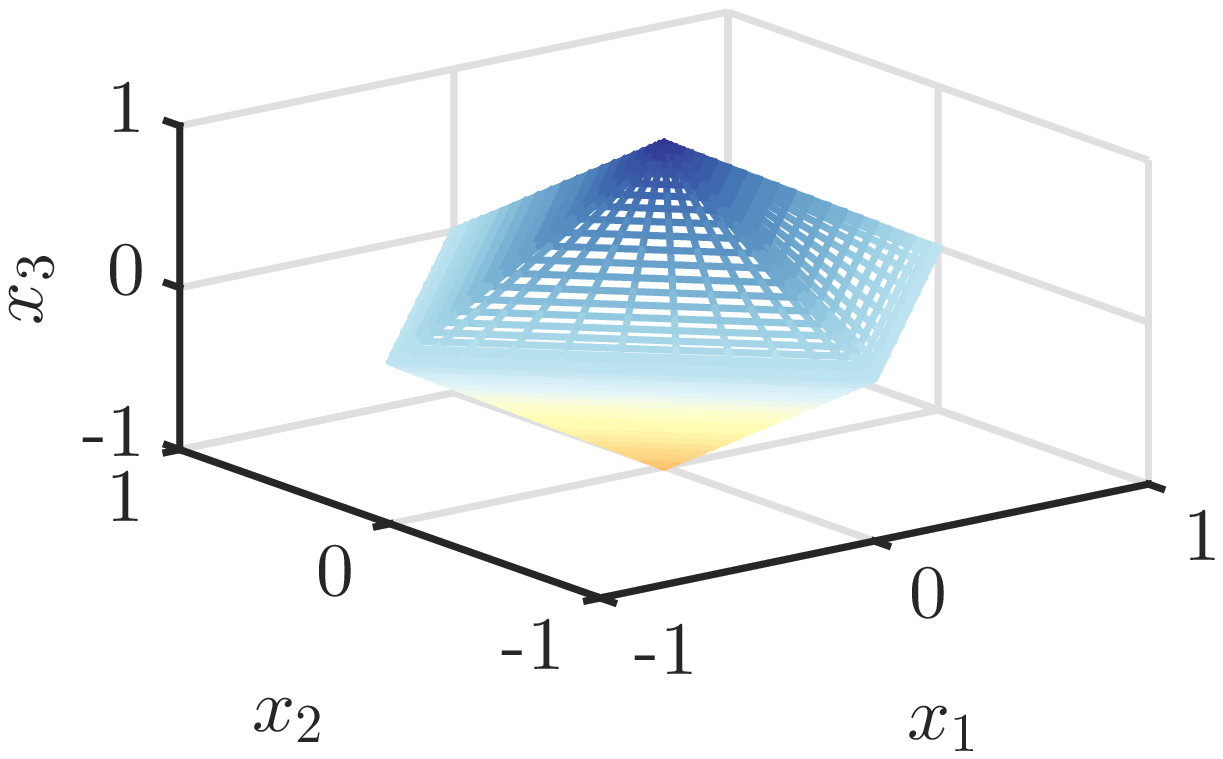}
}
\centering \subfigure[$p=0.5$.]{
  \includegraphics[width=0.492\linewidth]{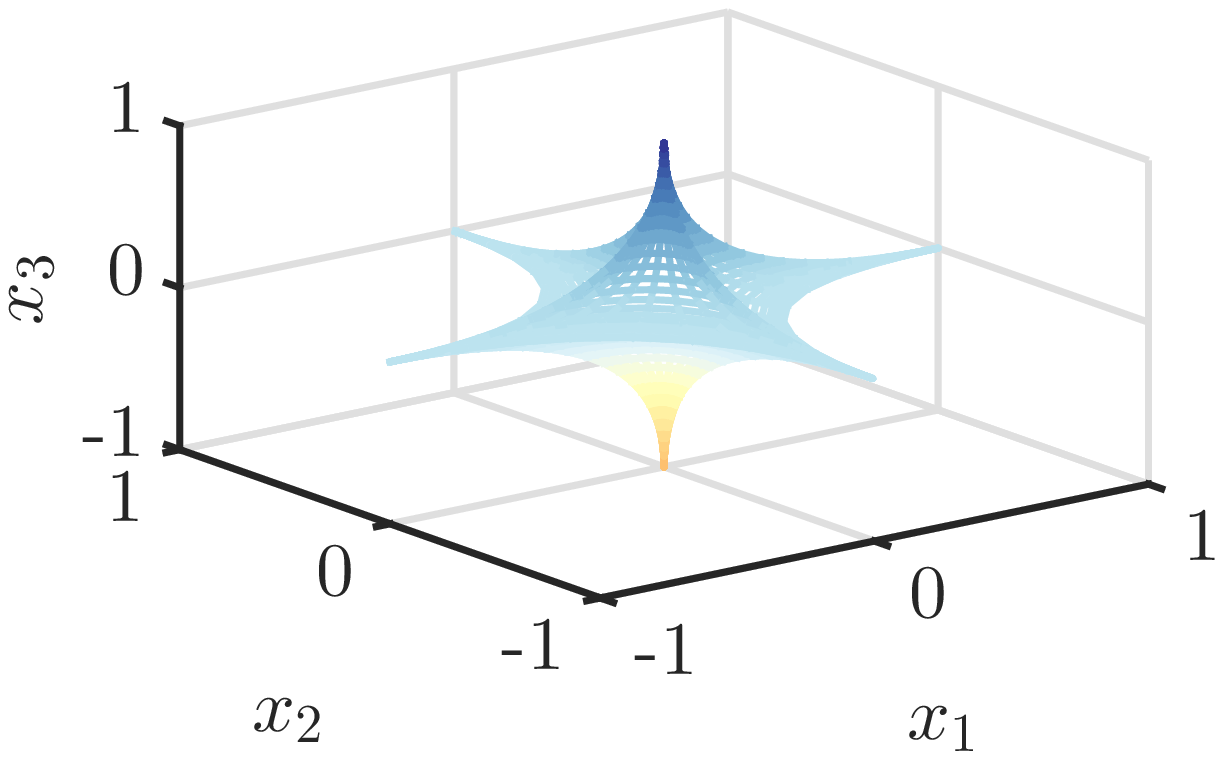}
} \hspace{-.55cm} \subfigure[$p=0.25$]{
  \includegraphics[width=0.492\linewidth]{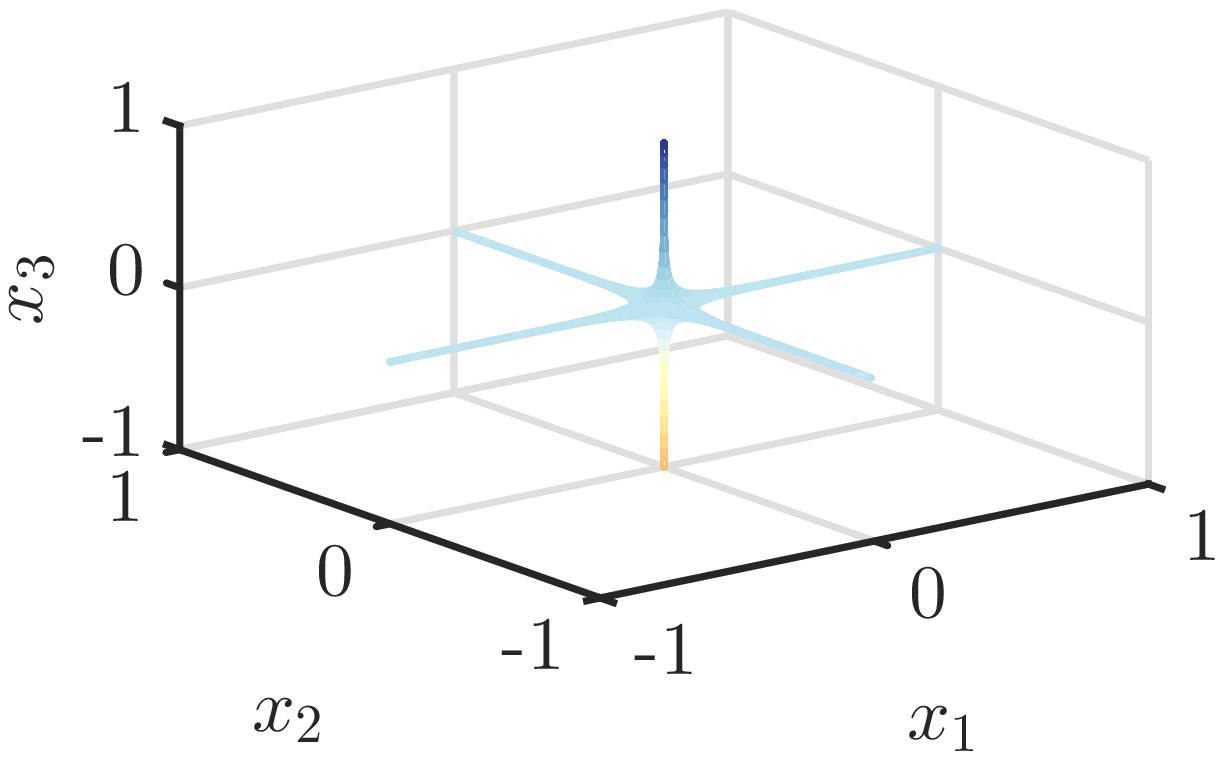}
}
\vspace{-9pt}
\caption{$\mc{B}_\bm{p}$ for various values of $\bm{p} =  p\cdot \bold{1}$.}
\label{fig:lpballs}
\end{figure}
To simplify the subsequent exposition, we study the case of a uniform distribution on $\mc{B}_\bm{p}$, and note that more general (generalized-radial) distributions are subject to future works. Therefore, in all that follows, the probability distribution $p_\bs{x}(\bm{x})$ is assumed to be %and measure for $\mc{A} \subset \bb{R}^N$ are given by (see \cite{BarGueMenNao05})
\begin{align}\label{equ:uni_lp}
%p_\bs{x}(\bm{x}) = 
%\begin{cases}
%\frac{1}{\rs{vol}( \mc{B}_\bm{p} )}  &\text{for } \bm{x} \in \mc{B}_\bm{p},\\ 
%0 &\text{else},
%\end{cases}
p_\bs{x}(\bm{x}) = \frac{1}{\rs{vol}( \mc{B}_\bm{p} )} \ds{1}_{\mc{B}_\bm{p}}(\bm{x}).
\end{align}
%and
%\begin{align}
%\bb{P}( \bm{x} \in \mc{A} ) := \frac{ \rs{vol}( \mc{A} \cap \mc{B}_p)}{ \rs{vol}( \mc{B}_p ) },
%\end{align}
%\vspace{-18pt}
For brevity, we use $\bs{x} \sim \mc{U}(\mc{B}_p)$ to refer to the
random variable $\bs{x}$ drawn according to
\eqref{equ:uni_lp}. Owing to the
lack of space, we omit an in-depth discussion of stochastic models and
refer an interested reader to the overview article on compressible
distributions in \cite{GrCeDa12} as well as the works on various
sparse L\'{e}vy processes in \cite{UnTa10}. 
Given the measurement model \eqref{equ:cs_meas}, we derive
the \textit{Bayesian mean-squared-error (MSE)} for a structured
nonlinear estimator composed of a linear operator followed by a Cartesian product of
univariate nonlinear mappings. For a recursive structure, a computationally much simpler
approach can be found in \cite{KamMan15} using a stochastic gradient
method. While this amounts to a better scalability w.r.t. the problem
dimension, the algorithm may converge slowly, or not at all, and
missing error estimates may restrict its applicability. For the case
of a polynomial mapping in canonical form, we analyze an alternating
optimization approach that is guaranteed to converge w.r.t. the
MSE objective. The latter is shown be to computable in
closed-form as a function of higher-order inner product statistics.
\begin{remark}[Bayesian vs. classical MMSE estimation]
  We highlight that the present paper targets the Bayesian MSE as
  opposed to classical MSE estimation. In the Bayesian setting, an
  optimal estimator in the sense of an average performance criterion
  is obtained under the assumption of a prior pdf of the estimand. As
  such, the optimal Bayesian estimator for the MSE criterion is given
  by the conditional mean, which is in general hard to obtain and is
  approximated in a hardware-efficient manner in this work. On the
  other hand, in classical MSE estimation, a certain realization of
  sparse vector $\bm{x}$ is chosen and an optimal estimator for the
  particular given case is sought. For the latter case, the optimal
  estimator is often not realizable due to its dependence on the
  particular realization $\bm{x}$ (see also \cite{Kay93}[Ch. 10] for
  additional illustrative examples).
\end{remark}

\subsection{Notation}\label{sec:not}
Scalar, vector and matrix random variables are denoted by lowercase,
bold lowercase and bold uppercase sans-serif letters $\ms{x}$,
$\bs{x}$, $\bs{X}$, while the corresponding realizations by serif
letters $x$, $\bm{x}$, $\bm{X}$. The sets of reals, nonnegative reals,
positive reals, nonnegative integers and natural numbers are
designated by $\bb{R}$, $\bb{R}_{+}$, $\bb{R}_{++}$, $\bb{N}_0$ and $\bb{N}$. We use $\bold{0}$, $\bold{1}$ and $\bm{I}$ to denote the
vectors of all zeros, all ones and the identity matrix, where the size
will be clear from the context. $\rs{tr}\{\cdot\}$,
$\rs{diag}(\bm{u})$, $(\cdot)^{\odot d}$ and
$\ds{1}_{\mc{X}}: \bm{x}\to \{0,1\}$ denote the trace of a matrix, the
diagonal matrix with elements of $\bm{u}$ on the diagonal, the
hadamard (i.e. entry-wise) power and the indicator function defined as
$\ds{1}_{\mc{X}}(\bm{x})=1$ if $\bm{x} \in \mc{X}$ and $0$
otherwise. $\mc{U}(\mc{X})$ is used to denote the
uniform distribution over the set $\mc{X}$,
$\bb{E}\left[ \cdot \right]$ is the expectation
operator and $\mc{B}_\bm{p}$ is the generalized unit ball defined in Lemma \ref{lem:volbp}.

\section{A primer for signals from $\mc{B}_\bm{p}$}\label{sec:lbpalls}
%\subsection{Monomial and higher-order inner product statistics}
Given the probability distribution in \eqref{equ:uni_lp} the first question is if $\rs{vol}(\mc{B}_p)$ can be obtained in closed form for general vectors $\bm{p}$ without using multivariate approximation techniques (e.g. cubature formulae) that are known to suffer from the so-called curse of dimensionality. It is interesting to note that an affirmative answer to this question can be traced back to works by Dirichlet on the Laplace transform \cite{Edw22} as was noted in \cite{Wan05} and appeared in different works from control theory to Banach space geometry (see \cite{CalDabTem98,BarGueMenNao05}). The respective result is restated in the following Lemma.
\begin{lemma}[Volume of generalized unit balls $\mc{B}_\bm{p}$]\label{lem:volbp}
Let $\bm{p} \in \bb{R}_{++}^N$, $\mc{B}_\bm{p}$ be given by $\mc{B}_\bm{p}:= \{\bm{x}: \sum_{n=1}^N \lvert x_n \rvert^{p_n} \leq 1\}\subset \bb{R}^N$ and
\begin{align}
\Gamma(z):= \int_0^{\infty} t^{z-1} \exp(-t) \, dt
\end{align}
denote the Gamma function (see \cite{Dav65} for a review of mathematical properties). Then, it holds that
\begin{align}
\rs{vol}(\mc{B}_\bm{p}) = \frac{2^N}{\prod_{n=1}^N p_n } \frac{ \prod_{n=1}^N \Gamma\left( \frac{1}{p_n} \right)}{\Gamma \left( 1+ \sum_{n=1}^N \frac{1}{p_n} \right) }.
\end{align}
\end{lemma}
\begin{proof}
The proof can be found e.g. in \cite{Wan05}.
\end{proof}
%\begin{lemma}[Volume of generalized unit balls $\mc{B}_\bm{p}$]\label{lem:volbp}
%Let $\bm{p} \in \bb{R}_{++}^N$ and $\mc{B}_\bm{p}$ be given by $\mc{B}_\bm{p}:= \{\bm{x}: \sum_{n=1}^N \lvert x_n \rvert^{p_n} \leq 1\}\subset \bb{R}^N$. Then,
%\begin{align}
%\rs{vol}(\mc{B}_\bm{p}) = \frac{2^N}{\prod_{n=1}^N p_n } \frac{ \prod_{n=1}^N \Gamma\left( \frac{1}{p_n} \right)}{\Gamma \left( 1+ \sum_{n=1}^N \frac{1}{p_n} \right) },
%\end{align}
%where 
%\begin{align}
%\Gamma(z):= \int_0^{\infty} t^{z-1} \exp(-t) \, dt
%\end{align}
%is the Gamma function (see \cite{Dav65} for a review of mathematical properties).
%\end{lemma}
%\begin{proof}
%The proof can be found e.g. in \cite{Wan05}.
%\end{proof}
As an extension of Lemma \ref{lem:volbp}, we obtain the following result for the integral as well as expectation of a monomial over $\mc{B}_\bm{p}$ w.r.t. to the measure \eqref{equ:uni_lp}, which forms the basis for the subsequent analysis.
\begin{lemma}[Expectation of monomials over $\mc{B}_\bm{p}$]\label{lem:mon_moment}
Let $\bm{x}^{\bg{\alpha}}$ denote the monomial ${x}_1^{\alpha_1} \cdots {x}_N^{\alpha_N}$ with $\bm{x}\in\bb{R}^N$ and $\bg{\alpha} \in \bb{N}_{0}^N$, $\bs{x} \sim \mc{U}(\mc{B}_\bm{p})$, and $2\bb{N}_0:=\{2\beta : \beta \in \bb{N}_0\}$ be the set of nonnegative even integers. Then, we have
\begin{align}\label{equ:int_mon}
\int_{\mc{B}_\bm{p}} \bm{x}^{\bg{\alpha}} \ d\bm{x} & = \begin{cases}
\frac{2^N}{\prod_{n=1}^N p_n } \frac{\prod_{n=1}^N \Gamma\left( \frac{\alpha_n + 1}{p_n} \right) }{\Gamma \left( 1 + \sum_{n=1}^N \frac{\alpha_n + 1}{p_n} \right) } &\text{for } \bg{\alpha} \in 2\bb{N}_0^N \\
0 &\text{otherwise}, \nonumber
\end{cases}
\end{align}
and
\begin{align}
\bb{E}_\bs{x} \left[ \bs{x}^{\bg{\alpha}} \right] = \frac{1}{\rs{vol}(\mc{B}_\bm{p})} \int_{\mc{B}_\bm{p}} \bm{x}^{\bg{\alpha}} \ d\bm{x}.
\end{align}
\end{lemma}
\begin{proof}
The proof is deferred to Appendix \ref{app:A}.
\end{proof}
Of course, $\rs{vol}(\mc{B}_\bm{p})$ can be obtained similarly as a special case of Lemma \ref{lem:mon_moment} using $\bg{\alpha}=\bold{0}$.
In the subsequent analysis, we also need to evaluate higher-order statistics of an inner product of $\bs{x}$ and some given $\bm{u}\in\bb{R}^N$, which is formalized in the following Lemma.
\begin{lemma}[Higher-order inner-product statistics]\label{lem:power_iprod}
Let $\bs{x}\sim\mc{U}(\mc{B}_\bm{p})$, $\bg{\alpha} \in \bb{N}_{0}^N$, $d \in \bb{N}_0$ and $\bm{u}\in\bb{R}^N$ be a given vector. Then, using
\begin{align}
{d \choose \bg{\alpha}} = \frac{d!}{\prod_{n=1}^N (\alpha_n !)},
\end{align} 
we obtain
\begin{align}
& \bb{E}_\bs{x} \left[ \langle \bm{u}, \bs{x} \rangle^d \right] = \sum\nolimits_{\lVert \bg{\alpha} \rVert_1 = d} {d \choose \bg{\alpha}}  \bm{u}^{\bg{\alpha}} \bb{E}_\bs{x} \left[ \bs{x}^{\bg{\alpha}} \right], \\
& \bb{E}_\bs{x} \left[ \ms{x}_i \langle \bm{u}, \bs{x} \rangle^d \right] = \sum\nolimits_{\lVert \bg{\alpha} \rVert_1 = d} {d \choose \bg{\alpha}}  \bm{u}^{\bg{\alpha}} \bb{E}_\bs{x} \left[ \bs{x}^{\bg{\alpha}+\bm{e}_i} \right], \\
& \bb{E}_\bs{x} \left[ \ms{x}_i \ms{x}_j \langle \bm{u}, \bs{x} \rangle^d \right]= \sum\nolimits_{\lVert \bg{\alpha} \rVert_1 = d} {d \choose \bg{\alpha}}  \bm{u}^{\bg{\alpha}} \bb{E}_\bs{x} \left[ \bs{x}^{\bg{\alpha}+\bm{e}_i + \bm{e}_j} \right],
\end{align}
where $\bm{e}_i$ denotes the $i$-th standard Euclidean basis vector in $\bb{R}^N$.
\end{lemma}
\begin{proof}
The Lemma follows from an application of the multinomial formula
\begin{align}
 \left( u_1 \ms{x}_1 + u_2 \ms{x}_2 + \hdots + u_N \ms{x}_N \right)^d =  \sum\nolimits_{\lVert \bg{\alpha} \rVert_1 = d} {d \choose \bg{\alpha}}  \bm{u}^{\bg{\alpha}} \bs{x}^{\bg{\alpha}} 
\end{align}
%with multinomial coefficient
%\begin{align}
%{d \choose \bg{\alpha}} = \frac{d!}{\prod_{n=1}^N (\alpha_n !)}
%\end{align}
together with the linearity of the expectation operator.
\end{proof}
In Fig. \ref{fig:sparse_process} we illustrate similarities and differences of various sparse processes that can be encountered in literature. The respective probability density functions are given in Tab. \ref{tab:sparse_process}. For a practical algorithm and implementation to generate signals from $\mc{B}_\bm{p}$ we refer the interested reader to \cite{CalDabTem98}, which was also used for the Monte-Carlo simulations in Sec. \ref{sec:sims}.\begin{table}
\centering
\begin{tabular}{ l  l }
\hline \hline
Model & $p_\bs{x}(\bm{x})$ \\
\hline
Gaussian & $\frac{1}{(2 \pi \sigma^2)^{\frac{N}{2}}} e^{ - \frac{ \lVert \bm{x} \rVert_2^2}{\sigma^2} }$ \\ \vspace{.05cm}
Laplace & $\left( \frac{\lambda}{2} \right)^N e^{ - \lambda \lVert \bm{x} \rVert_1 }$ \\ \vspace{-.1cm}
Compound Poisson, & $ \prod_{n=1}^N \Big(e^{-\lambda} \delta(x_n) +  $ \\ 
Gaussian amplitude & \ $+( 1-e^{-\lambda})  \frac{1}{\sqrt{2 \pi \sigma^2}} \exp \left( \frac{ - x_n^2}{2 \sigma^2 } \right) \Big)$\\ \vspace{0.1cm}
uniform $\mc{B}_\bm{p}$ & $\frac{1}{\rs{vol}( \mc{B}_\bm{p} )} \ds{1}_{\mc{B}_\bm{p}}(\bm{x})$ \\
\hline
\vspace{0.01cm}
\end{tabular}
\vspace{-9pt}
\caption{PDFs of various (sparse) processes.}
\label{tab:sparse_process}
\end{table}
\section{Bayesian estimators for signals from $\mc{B}_\bm{p}$}
\subsection{MAP estimation}
We start this section with a brief review of general Bayesian estimators following a standard textbook in the field \cite{Kay93}.
\begin{definition}[MAP estimator]\label{def:map}
Let $\bm{y}$ be defined by \eqref{equ:cs_meas} and $p_\bs{x}(\bm{x})$ be given by \eqref{equ:uni_lp}. A MAP estimate 
\begin{align}
\bh{x}_\rs{map} \in {\rs{argmax}_\bm{x}} \ p_{\bs{y} \vert \bs{x}}(\bm{y}\vert \bm{x}) p_\bs{x}(\bm{x})
\end{align} 
is given by
\begin{align}\label{equ:map1}
\bh{x}_\rs{map} \in {\rs{argmax}_\bm{x}} \ \delta(\bm{y} - \bm{A}\bm{x}) p_{\bs{x}}(\bm{x}).
\end{align}
Here, $\delta(\bm{z})$ denotes the idealized dirac-delta point mass at $\bm{z}=\bold{0}$. We note that \eqref{equ:map1} can be equivalently written as 
\begin{align}
\bh{x}_\rs{map} \in \left( \bm{x}_0 + \rs{null}(\bm{A}) \right) \cap \mc{B}_\bm{p},
\end{align}
where $\bm{x}_0$ is an arbitrary point satisfying $\bm{y} = \bm{A} \bm{x}_0$.
\end{definition}
The MAP estimator provides an excellent estimation performance, but it usually amounts to solving a costly optimization problem rendering it infeasible for most real-time applications. Some relevant examples of such applications  in the field of communications include sparse channel estimation \cite{RajBhaCavAaz02} and sparse multiuser detection \cite{NikYiBayAu14}, where the interest is in the development of dedicated chips based on integrated circuit (IC) architectures that exploit pipelining as well as parallelism. In such settings, even a seemingly simple matrix-inverse is usually avoided as it scales cubic in the number of inputs \cite{RajBhaCavAaz02}. 
\subsection{Linear Bayesian MMSE estimation}
We proceed with low-complexity linear Bayesian MMSE (LMMSE) estimators that, whilst being inferior to the MAP in terms of estimation performance, may be easily implemented and often offer acceptable performance guarantees.
\begin{definition}[Linear Bayesian MMSE]\label{def:lmmse}
Let $\bm{y}$ and $p_\bs{x}(\bm{x})$ be given by \eqref{equ:cs_meas} and \eqref{equ:uni_lp}. The linear Bayesian MMSE estimator $\bm{W}_{\rs{lmmse}}$ 
is the solution to
\begin{align}\label{equ:lmmse}
\bm{W}_{\rs{lmmse}} \in \underset{\bm{W} \in \bb{R}^{N \times M}}{\rs{argmin}} \ \bb{E}_\bs{x} \left[ \lVert \bs{x} - \bm{W}\bm{A}\bs{x} \rVert_2^2 \right],
\end{align}
where the expectation is taken w.r.t. $\bs{x}\sim\mc{U}(\mc{B}_\bm{p})$.
\end{definition}
\begin{theorem}[Linear Bayesian MMSE]\label{th:lmmse}
Let $\bm{y}$ be given by \eqref{equ:cs_meas} and $p_\bs{x}(\bm{x})$ by \eqref{equ:uni_lp}. Assuming that the inverse exists, the optimal linear estimator according to Def. \ref{def:lmmse} can be obtained by
\begin{align}
\bm{W}_\rs{lmmse} = \bm{C}_{\bs{x}} \bm{A}^T \left( \bm{A} \bm{C}_{\bs{x}} \bm{A}^T \right)^{-1}
\end{align}
with $\bm{C}_{\bs{x}} := \bb{E}\left[ \bs{x} \bs{x}^T \right]$ given by
\begin{align}
[\bm{C}_{\bs{x}}]_{i,j} := \bb{E}_\bs{x} \left[ \bs{x}^{\bm{e}_i + \bm{e}_j} \right].
\end{align}
\end{theorem}
\begin{proof}
The proof is a standard results in Bayesian MMSE estimation (see e.g. \cite[p. 364]{Kay93}).
\end{proof}
\begin{figure}
\centering \subfigure[Gaussian distribution (not sparse) with $\mu=0$, $\sigma=1$.]{
  \includegraphics[width=0.475\linewidth]{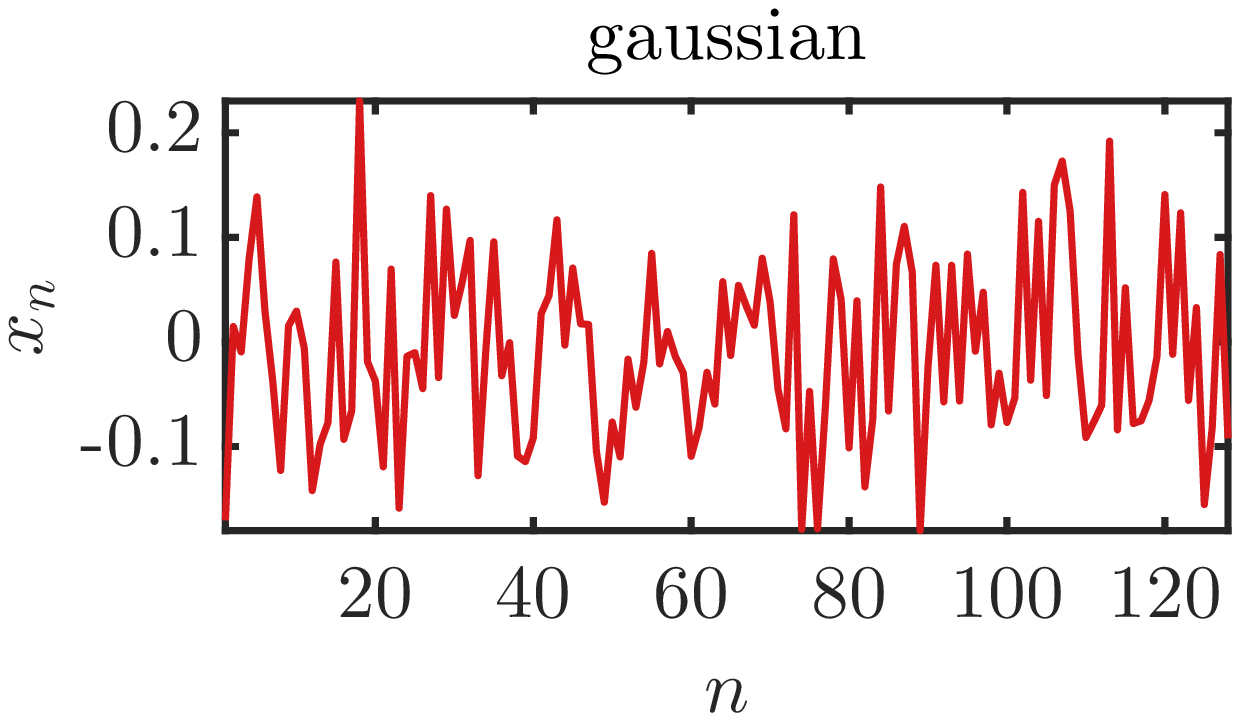}
} \hspace{-.05cm} \subfigure[Laplace distribution with $\lambda=5$.]{
  \includegraphics[width=0.475\linewidth]{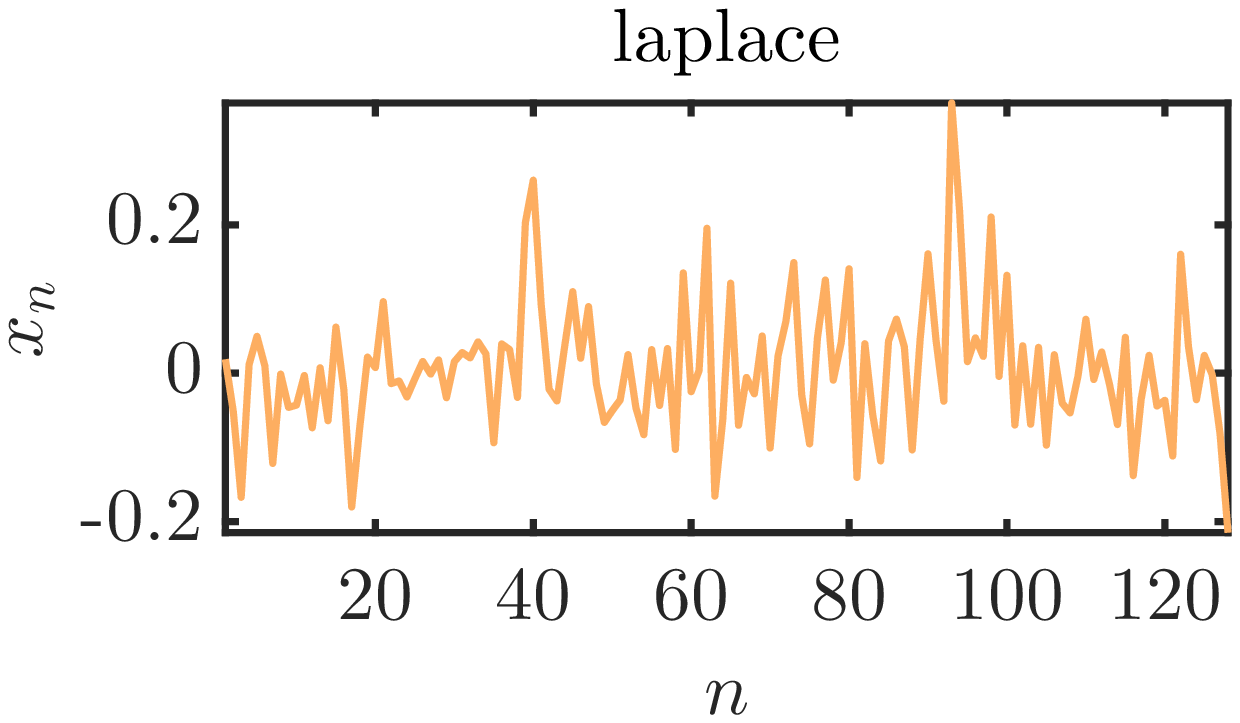}
}
\centering \subfigure[Compount poisson distribution with $\lambda=0.25$, $\mu=0$, $\sigma=1$.]{
  \includegraphics[width=0.475\linewidth]{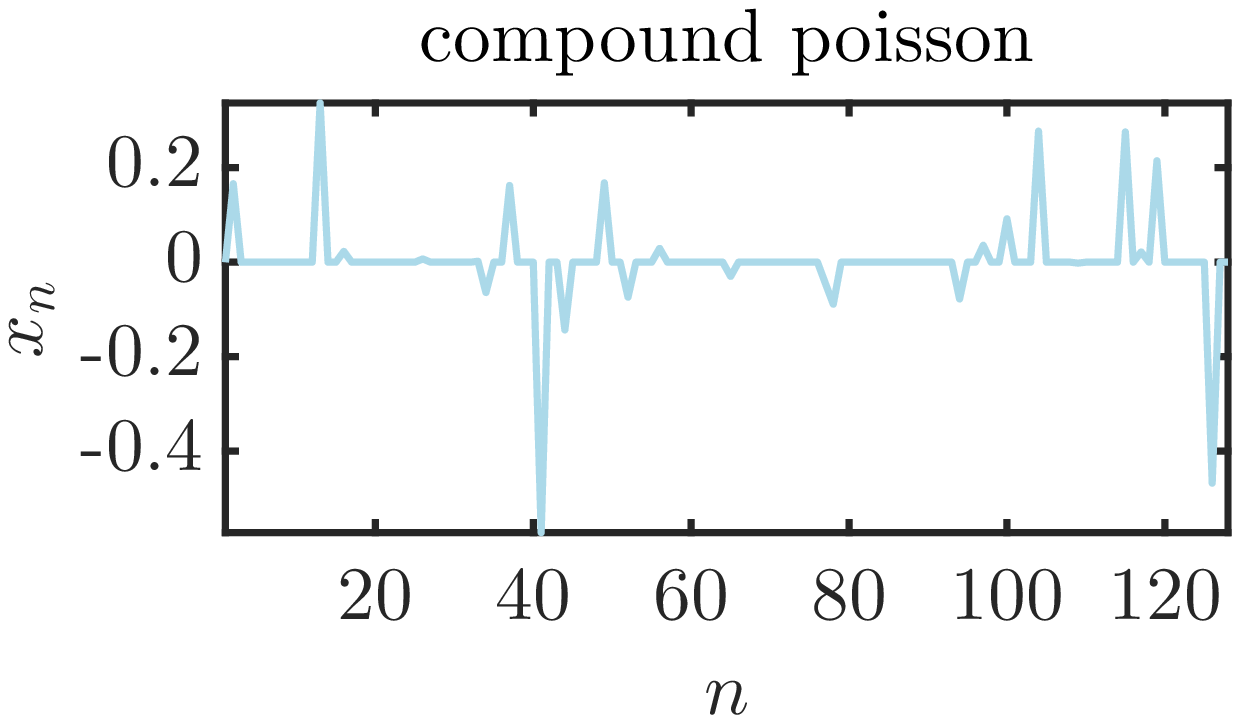}
} \hspace{-.05cm} \subfigure[Uniform $\mc{B}_\bm{p}$ with $\bm{p}=0.33 \cdot \bold{1}$.]{
  \includegraphics[width=0.475\linewidth]{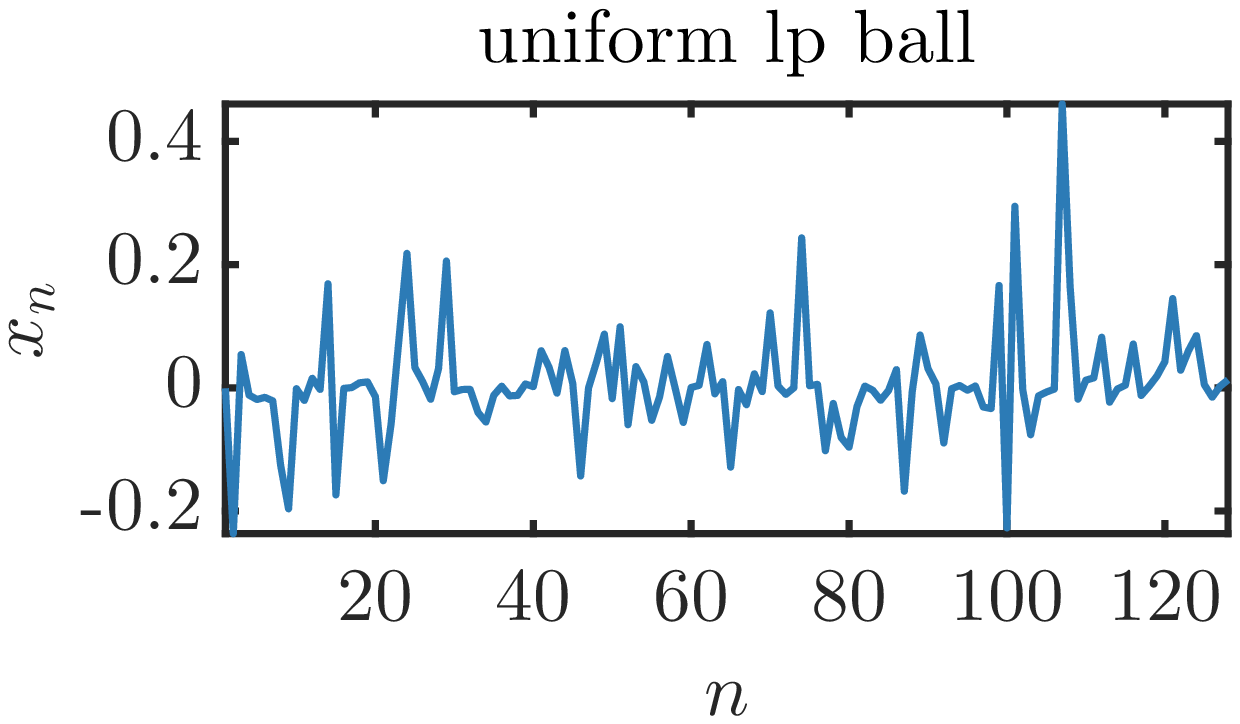}
}
\vspace{-9pt}
\caption{Realizations of various (sparse) processes in $\bb{R}^{128}$. Signals are normalized to unit $\ell_2$-norm.}
\label{fig:sparse_process}
\end{figure}
The corresponding Bayesian MSE is given by
\begin{align}\label{equ:smse}
\varepsilon_\rs{lmse}(\bm{W}) = & \bb{E}_\bs{x} \left[ \lVert \bs{x} - \bm{W} \bm{A} \bs{x} \rVert_2^2 \right] = \rs{tr} \left\{\bm{C}_{\bs{x}} \right\} - 2 \rs{tr} \left\{ \bm{W} \bm{A} \bm{C}_{\bs{x}} \right\} \nonumber \\ & + \rs{tr} \left\{ \bm{A}^T \bm{W}^T \bm{W} \bm{A} \bm{C}_{\bs{x}} \right\}.
\end{align}
\vspace{-18pt}
\subsection{Structured nonlinear estimation}
An increasingly popular technique for recovering sparse signals consists in using a linear mapping followed by a Cartesian product of univariate nonlinearities (e.g. classical or learned \textit{iterative shrinkage-thresholding} algorithms \cite{BecTeb09},\cite{KamMan15}) in an alternating fashion. As a conceptual analogue, we propose a nonlinear Bayesian MMSE estimator using a similar structural assumption.
\begin{proposition}[Structured nonlinear MMSE estimator]\label{prop:smmse}
Let $\mc{T} := \mc{T}_1 \times \hdots \times \mc{T}_N : \bb{R} \times \hdots \times \bb{R} \mapsto  \bb{R} \times \hdots \times \bb{R}$ be a Cartesian product of univariate nonlinear mappings and define the sructured Bayesian MMSE (SMMSE) estimator to be of the form
\begin{align}
\bh{x} = \mc{T}\left( \bm{W} \bm{y} \right) = \mc{T}\left( \bm{W} \bm{A} \bm{x} \right),
\end{align}
where for ease of practical realization we further impose equality among the nonlinear mappings, i.e., $\mc{T}_1 = \hdots = \mc{T}_N$.
\end{proposition}
An illustration of the SMMSE estimators' structure is shown in Fig. \ref{fig:smmse_nnet}.\footnote{The limitation to one linear and one nonlinear Cartesian product mapping with presumed identity is linked to the resulting computational complexity and may be overcome by appropriate approximation techniques.}% In addition, equality among nonlinear mappings can be exploited to derive area-efficient VLSI implementations, as allocated LUTs can be reused to compute $\mc{T}_1,\hdots,\mc{T}_N$ in subsequent clock cycles.}
\begin{figure}[htb]
\centering
 \psfrag{s}[bc][bc]{$\bm{W}$}
 \psfrag{y1}[bc][bc]{${y}_1$}
 \psfrag{yK}[bc][bc]{${y}_M$}
 \psfrag{x1}[bc][bc]{${x}_1$}
 \psfrag{xK}[bc][bc]{${x}_N$}
 \psfrag{t1}[bc][bc]{${\hat{x}}_1$}
 \psfrag{tK}[bc][bc]{${\hat{x}}_N$}
 \psfrag{A}[bc][bc]{$\bm{A}$}
 \psfrag{W}[bc][bc]{$\bm{W}$}
 \psfrag{T}[bc][bc]{$\mc{T}$}
  \includegraphics[width= 0.95\linewidth]{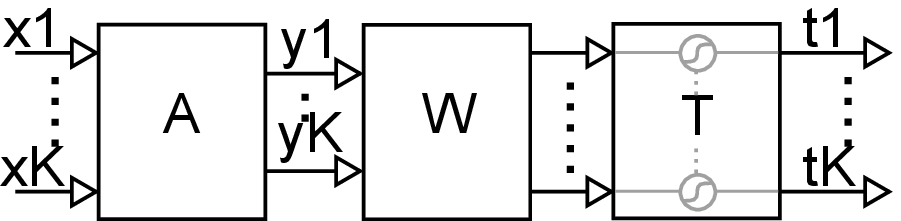}
  \vspace{-9pt}
 \caption{Structured nonlinear Bayesian MMSE estimator composed of a linear map $\bm{W}$ and a Cartesian product of univariate nonlinear maps $\mc{T}:=[\mc{T}_1 \cdots \mc{T}_N]^T$.}
\label{fig:smmse_nnet}
\end{figure}
In this paper, we analyze explicitly the canonical polynomial map 
\begin{align} 
\mc{T}_i(t):= \sum\nolimits_{d=0}^D a_d t^d
\end{align}
resulting in an estimate
\begin{align}
\bh{x} = \sum\nolimits_{d=0}^D a_d \left( \bm{W} \bm{A} \bm{x} \right)^{\odot d}.
\end{align}
Accordingly, the corresponding Bayesian MSE is given by
%\begin{align}\label{equ:smse1}
%& \varepsilon_\rs{smse}(\bm{a},\bm{W}) =  \bb{E}_\bs{x} \left[ \biggl\lVert \bs{x} - \sum_{d=0}^D a_d \left( \bm{W} \bm{A} \bs{x} \right)^{\odot d} \biggr\rVert_2^2 \right] \\
%& = \rs{tr} \left\{\bm{C}_{\bs{x}} \right\} - 2 \rs{tr} \underbrace{\left\{ \bb{E}_\bs{x} \left[ \sum_{d=0}^D a_d \rs{diag}(\bs{x}) \rs{diag}^d\left( \bm{W} \bm{A} \bs{x} \right)  \right] \right\}}_{\bm{C}_{\bs{x}\bs{\hat{x}}}} \nonumber \\
%& + \rs{tr} \underbrace{\left\{ \bb{E}_\bs{x} \left[ \sum_{i=0}^D \sum_{j=0}^D a_i a_j \rs{diag}^{i+j} \left( \bm{W} \bm{A} \bs{x} \right) \right] \right\}}_{\bm{C}_{\bs{\hat{x}}}} \nonumber ,
%\end{align}
\begin{align}\label{equ:smse1}
\varepsilon_\rs{smse}(\bm{a},\bm{W}) & =  \bb{E}_\bs{x} \left[ \lVert \bs{x} - \bs{\hat{x}} \rVert_2^2 \right] \nonumber \\
& = \rs{tr} \left\{\bm{C}_{\bs{x}} \right\} - 2 \rs{tr} \left\{\bm{C}_{\bs{x}\bs{\hat{x}}} \right\} + \rs{tr} \left\{\bm{C}_{\bs{\hat{x}}} \right\},
\end{align}
where $\rs{tr} \left\{\bm{C}_{\bs{x}} \right\}$ follows from Th. \ref{th:lmmse}. The two other terms are equal to
\begin{align}
\rs{tr} \left\{\bm{C}_{\bs{x}\bs{\hat{x}}} \right\} & = \rs{tr} \left\{ \bb{E}_\bs{x} \left[ \sum_{d=0}^D a_d \rs{diag}^d\left( \bm{W} \bm{A} \bs{x} \right) \rs{diag}(\bs{x})  \right] \right\} \label{equ:tr11}   \\
& = \bb{E}_\bs{x} \left[ \rs{diag}(\bs{x}) \bm{V} \right] \bm{a} \label{equ:tr12} \\
\rs{tr} \left\{\bm{C}_{\bs{\hat{x}}} \right\} & = \rs{tr} \left\{ \bb{E}_\bs{x} \left[ \sum_{i=0}^D \sum_{j=0}^D a_i a_j \rs{diag}^{i+j} \left( \bm{W} \bm{A} \bs{x} \right) \right] \right\} \label{equ:tr21}  \\
& = \bm{a}^T \bb{E}_\bs{x}\left[ \bm{V}^T \bm{V} \right] \bm{a} \label{equ:tr22},
\end{align}
where we use the convention that $\rs{diag}^0(\bm{u}) = \bm{I}$ and define the Vandermonde matrix $\bm{V}$ as
\begin{align}\label{equ:vand}
\bm{V} := \left[ \bold{1}, (\bm{W} \bm{A}\bs{x})^{\odot 1}, \hdots , (\bm{W} \bm{A}\bs{x})^{\odot D} \right].
\end{align}
To obtain \eqref{equ:tr11}-\eqref{equ:tr22} we define $\bm{U}=[\bm{u}_1,\hdots,\bm{u}_n]^T:=\bm{W}\bm{A}$ and apply Lemma \ref{lem:power_iprod} to compute the required expectations entrywise:
\begin{align}
&\left\{ \bb{E}_\bs{x} \left[ \bs{x}^T \bm{V} \right] \right\}_{i,j}  = \bb{E}_\bs{x} \left[ \ms{x}_i \langle \bm{u}_i,\bs{x} \rangle^{j-1} \right] \\
&\qquad \qquad \ \forall \{i,j\}  \in \{1,\hdots,N\}\times\{1,\hdots,D+1\}, \nonumber \\
&\left\{ \bb{E}_\bs{x} \left[ \bm{V}^T \bm{V} \right] \right\}_{i,j}  = \sum_{n=1}^N \bb{E}_\bs{x} \left[ \langle \bm{u}_n,\bs{x} \rangle^{i+j-2} \right] \\
& \qquad \qquad \ \forall \{i,j\} \in \{1,\hdots,N\}^2. \nonumber
\end{align}
\section{Alternating minimization of the SMSE}
The aim of this section is to derive an algorithmic solution to the minimization of the Bayesian SMSE \eqref{equ:smse1}, i.e., solving (approximately) the problem
\begin{align}\label{equ:ex_opt}
\underset{ \substack{\bm{a}\in\bb{R}^{D+1} \\ \bm{W} \in \bb{R}^{N \times M}}}{\min} \ \varepsilon_\rs{smse}(\bm{a},\bm{W}).
\end{align}
The reader should note that for $\bm{a}:=\bm{e}_2 \in \bb{R}^{D+1}$ the problem reduces to the LMMSE setting from Th. \ref{th:lmmse}. As such, the LMMSE estimator is a particular instance of the SMMSE estimator, and therefore it yields an upper bound on the achievable MSE. On the other hand, the integrand (expectation) in \eqref{equ:smse1} is nonnegative for every $\bm{x} \in \mc{B}_\bm{p}$. Hence, we can write
\begin{align}\label{equ:smse_bounds}
0 \leq \underset{ \substack{\bm{a}\in\bb{R}^{D+1} \\ \bm{W} \in \bb{R}^{N \times M}}}{\min} \ \varepsilon_\rs{smse}(\bm{a},\bm{W}) \leq \underset{ \bm{W} \in \bb{R}^{N \times M}}{\min} \ \varepsilon_\rs{lmse}(\bm{W}).
\end{align}
A widely-used algorithm for optimization problems with block partitioned arguments is the alternating minimization algorithm (AMA) \cite{GriSci00}, which is also often referred to as block coordinate descent method \cite{Ber99}, given in Alg. \ref{alg:altopt} for Problem \eqref{equ:ex_opt}.
\begin{algorithm}[h]\label{alg:altopt}
\SetKwInput{KwData}{\textbf{Input}}
\SetKwInput{KwResult}{\textbf{Output}}
\KwData{$\bm{W}_\star^{(0)}$, $\bm{a}_\star^{(0)}$}
\KwResult{$\bm{W}^\star$, $\bm{a}^\star$}
\For{$k=0,1,\hdots$}{
\begin{subequations}\label{equ:altopt}
\begin{align}
\bm{a}_\star^{(k+1)}  & \in \  \underset{ \bm{a} \in \bb{R}^{D+1}}{\rs{argmin}} \  \varepsilon_\rs{smse}(\bm{a},\bm{W}_\star^{(k)}) \label{equ:altopt.1} \\
\bm{W}_\star^{(k+1)}  & \in \  \underset{ \bm{W} \in \bb{R}^{N \times N}}{\rs{argmin}} \  \varepsilon_\rs{smse}(\bm{a}_\star^{(k+1)},\bm{W}) \label{equ:altopt.2}
\end{align}
\end{subequations}
}
\caption{Alternating minimization algorithm.}
\end{algorithm}
\vspace{-9pt}
The algorithm generates a non-increasing sequence of objective values since
\begin{align}
\forall k \in \bb{N}_0: \ \varepsilon_\rs{smse}(\bm{a}_\star^{(k)},\bm{W}_\star^{(k)}) & \geq \varepsilon_\rs{smse}(\bm{a}_\star^{(k+1)},\bm{W}_\star^{(k)}) \geq \\ & \geq \varepsilon_\rs{smse}(\bm{a}_\star^{(k+1)},\bm{W}_\star^{(k+1)}).
\end{align}
Due to the monotone convergence theorem a first consequence is that Alg. \ref{alg:altopt} converges w.r.t. the MSE objective, since by \eqref{equ:smse_bounds} the objective function is bounded from below. It was shown in \cite{GriSci00} that in convex as well as non-convex settings the generated sequence of solutions $(\bm{a}_\star^{(k)},\bm{W}_\star^{(k)})$ converges to a critical point of problem \eqref{equ:ex_opt} (provided that the generated sequence admits limit points and for each subproblem of Alg. \ref{alg:altopt} the minimum is uniquely attained). The latter non-convex setting indeed applies to Problem \eqref{equ:ex_opt} as can be seen from the optimization variable $\bm{W}$ being the argument of a generally non-convex polynomial map.

We highlight that from a numerical viewpoint the aforementioned convergence to critical points {may not be guaranteed (i.e. we may suffice ourselves with monotone convergence w.r.t. the MSE objective)} since the assumption that \emph{optimal} solutions to every subproblem \eqref{equ:altopt.2} of Alg. \ref{alg:altopt} can be computed is usually violated. For an accompanying numerical implementation of Alg. \ref{alg:altopt}, we first note that if the matrix $\bb{E}_\bs{x} \big[ \bm{V}^{T,(k)} \bm{V}^{(k)} \big]$ is positive-definite,\footnote{We strongly conjecture that this matrix is positive definite. The conjecture is based on extensive numerical simulations. Although a formal proof is missing, the conjecture is assumed to be valid in what follows.} then the first subproblem \eqref{equ:altopt.1} is strictly convex and admits a closed form solution by exploiting the first-order optimality condition
\begin{align}
\frac{\partial}{\partial \bm{a}} \varepsilon & := \begin{bmatrix} \frac{\partial \varepsilon}{\partial a_0} & \cdots & \frac{\partial \varepsilon}{\partial a_D} \end{bmatrix}^T  \\
& = -2 \bb{E}_\bs{x} \left[ \bm{V}^{(k),T} \bs{x} \right] + 2 \bb{E}_\bs{x} \left[ \bm{V}^{(k),T} \bm{V}^{(k)} \right] \bm{a} \stackrel{!}{=} \bold{0} .
\end{align}
Using \eqref{equ:vand} for some given $\bm{W}^{(k)}$ we obtain a numerical solution
\begin{align}\label{equ:upd_a}
\bm{a}_\diamond^{(k+1)} :=  \bb{E}_\bs{x} \left[ \bm{V}^{T,(k)} \bm{V}^{(k)} \right]^{-1} \bb{E}_\bs{x} \left[ \bm{V}^{T,(k)} \bs{x} \right].
\end{align}
For the generally non-convex subproblem \eqref{equ:altopt.2} we propose a numerical implementation based on a simple steepest-descent iteration to find a critical point $\bm{W}_\diamond^{(k)}$ as an approximation to the optimal solution $\bm{W}_\star^{(k)}$ using the following result for the partial derivative defined as
\begin{align}
\frac{\partial}{\partial \bm{W}} \varepsilon  := \begin{bmatrix} \frac{\partial \varepsilon}{\partial {W}_{1,1}} & \cdots & \frac{\partial \varepsilon}{\partial W_{1,M}}  \\
\vdots & \ddots & \vdots \\
\frac{\partial \varepsilon}{\partial {W}_{N,1}} & \cdots & \frac{\partial \varepsilon}{\partial {W}_{N,M}} \end{bmatrix}.
\end{align}
\begin{proposition}\label{prop:gradients}
Let $\rs{tr}\left\{\bm{C}_{\bs{x}\bs{\hat{x}}}\right\}$ and $\rs{tr}\left\{\bm{C}_{\bs{\hat{x}}}\right\}$ be given by \eqref{equ:tr11} and \eqref{equ:tr21}. Then, it holds that
\begin{align}\label{equ:gradw1}
\frac{\partial}{\partial \bm{W}} \rs{tr}\left\{\bm{C}_{\bs{x}\bs{\hat{x}}}\right\} = \bb{E}_\bs{x} \left[ \sum\nolimits_{d=1}^D d a_d \rs{diag}^{d-1}(\bm{W}\bm{A}\bs{x}) \bs{x}\bs{x}^T \bm{A} \right]
\end{align}
and
\begin{align}\label{equ:gradw2}
& \frac{\partial}{\partial \bm{W}} \rs{tr}\left\{\bm{C}_{\bs{\hat{x}}}\right\} = \\
&= \bb{E}_\bs{x} \left[ \sum_{k=0}^D \sum_{\substack{l=0 \\ [k,l]\neq \bold{0}}}^D (k+l)a_k a_l \rs{diag}^{k+l-1}(\bm{W}\bm{A}\bs{x}) \bold{1} \bs{x}^T \bm{A}^T \right]. \nonumber
\end{align}
\end{proposition}
\begin{proof}
The proof is deferred to Appendix \ref{app:B}.
\end{proof}
To compute the expectations in Prop. \ref{prop:gradients} we use Lemma \ref{lem:power_iprod} and evaluate the matrix numerically to obtain
\begin{align}
& \left[ \frac{\partial}{\partial \bm{W}} \rs{tr}\left\{\bm{C}_{\bs{x}\bs{\hat{x}}}\right\} \right]_{i,j} = \sum_{d=1}^D d a_d \bb{E}_\bs{x} \left[ \ms{x}_i\ms{x}_j \langle \bm{u}_i, \bs{x} \rangle^{d-1} \right] , \\
& \left[ \frac{\partial}{\partial \bm{W}} \rs{tr}\left\{\bm{C}_{\bs{\hat{x}}}\right\} \right]_{i,j} = \sum_{k=0}^D \sum_{\substack{l=0 \\ [k,l]\neq \bold{0} }}^D (k+l)a_k a_l \bb{E}_\bs{x} \left[ \ms{x}_j  \langle \bm{u}_i, \bs{x} \rangle^{i+j-1} \right]  \nonumber
\end{align}
$\forall \{i,j\} \in \{1,\hdots,N\}^2$. A further description of the numerical implementation is provided in the following Section.

\section{Numerical Results}\label{sec:sims}
To obtain the proposed structured Bayesian MMSE estimator, we solve the optimization problem \eqref{equ:altopt} using the update \eqref{equ:upd_a} for \eqref{equ:altopt.1} and a reference implementation of the steepest-descent algorithm with Armijo line-search \cite{BouMis13} using the gradients \eqref{equ:gradw1}, \eqref{equ:gradw2} for \eqref{equ:altopt.2}. 
We evaluate the normalized MSE defined as
\begin{align}
\rs{NMSE}:= \varepsilon(\bm{a}_\diamond,\bm{W}_\diamond)/ \rs{tr}(\bm{C}_\bs{x})
\end{align}
for a set of structurally different sensing matrices $\bm{A}\in \bb{R}^{3 \times 6}$ given as
\begin{enumerate}
\item an \textit{equiangular tight frame} (i.e. $\bm{A}_1:=[\bm{a}_1,\hdots,\bm{a}_6]$ s.t. $\lVert \bm{a}_i \rVert_2 = 1 \ \forall i$ and $\lvert \langle \bm{a}_i,\bm{a}_j \rangle \rvert = {\sqrt{N-M}}/{\sqrt{M(N-1)}}$ $\forall i\neq j$), 
\item a subsampled orthogonal matrix $\bm{A}_2$ (with $\bm{A}_2 \bm{A}_2^T = \bm{I}$), and 
\item a random matrix generated by drawing i.i.d. Gaussian entries followed by a normalization of rows.
\end{enumerate}
The remaining parameters are $\bm{p} := p \cdot \bold{1}$ with $p\in [0.4,2]$ and the polynomial map is set to degree $D=9$. As initial values we use $\bm{a}_\diamond^{(0)}=\bold{0}$ and a scaled Moore-Penrose pseudo-inverse $\bm{W}_\diamond^{(0)}=c \bm{A}^{\dagger}$, with scaling set to $c=10$ to stabilize the polynomial map, that were found experimentally. The results in terms of the $\rs{NMSE}$ are shown in Fig. \ref{fig:sim_results} and in terms of the optimal nonlinearities of the polynomial map for $\bm{A}_1$ in Fig. \ref{fig:sim_pcoeffs_etf}. For comparison, we also show the results for $\ell_1$-minimization (i.e. $\bm{\hat{x}} \in \rs{argmin}_{\bm{A}\bm{x}=\bm{y}} \bm{x}$) for $p\leq 1$ which were obtained using CVX \cite{GraBoy10}. We note that for this case $\ell_1$-minimization yields an interior point in the convex-hull $\mc{B}_\bold{1} \supseteq \mc{B}_{\bm{p}\leq \bold{1}}$ which should be a good approximation of the MAP estimate \eqref{equ:map1}.
Due to the high complexity of obtaining the optimal numerical parameters $(\bm{a}_\diamond,\bm{W}_\diamond)$ of the structured Bayesian MMSE estimator using the described numerical approximation of Alg. \ref{alg:altopt}, we limit our analysis to the low-dimensional setting and defer the high-dimensional analysis to a future study using e.g. faster approximate methods. We note, that the upper bound $0.5$ of the $\rs{NMSE}$ results from the compression factor $M/N$. It is interesting to see that the nonlinear Bayesian MMSE estimator in conjunction with the equiangular tight frame $\bm{A}_1$ resulted in the highest performance gains, with an approximate performance increase of (i) $20\,\%$, (ii) $15\,\%$ and (iii) $13\,\%$ over (i) the linear estimator (independent of the mapping $\bm{A}$), (ii) the subsampled orthogonal matrix and (iii) the normalized i.i.d. matrix. The optimization to obtain the SMMSE estimator for the predetermined set of sensing matrices and characteristic vectors was performed offline using an Amazon AWS {c4.8xlarge} instance and $36$ parallel threads. In terms of complexity, the estimation of $\bm{\hat{x}}$ given $\bm{A}_{i\in\{1,2,3\}} \bm{x}$ by the SMMSE estimator was observed to be more then a thousand-fold faster than $\ell_1$-minimization on a laptop with i7-2.9 GHz processor.\footnote{In the spirit of reproducible research, the simulation code used to generate the figures is available at \url{https://github.com/stli/MLSP2016_OptNonlin}.}
\vspace{-9pt}
\begin{figure}[htb]
\centering
  \includegraphics[width= .93\linewidth]{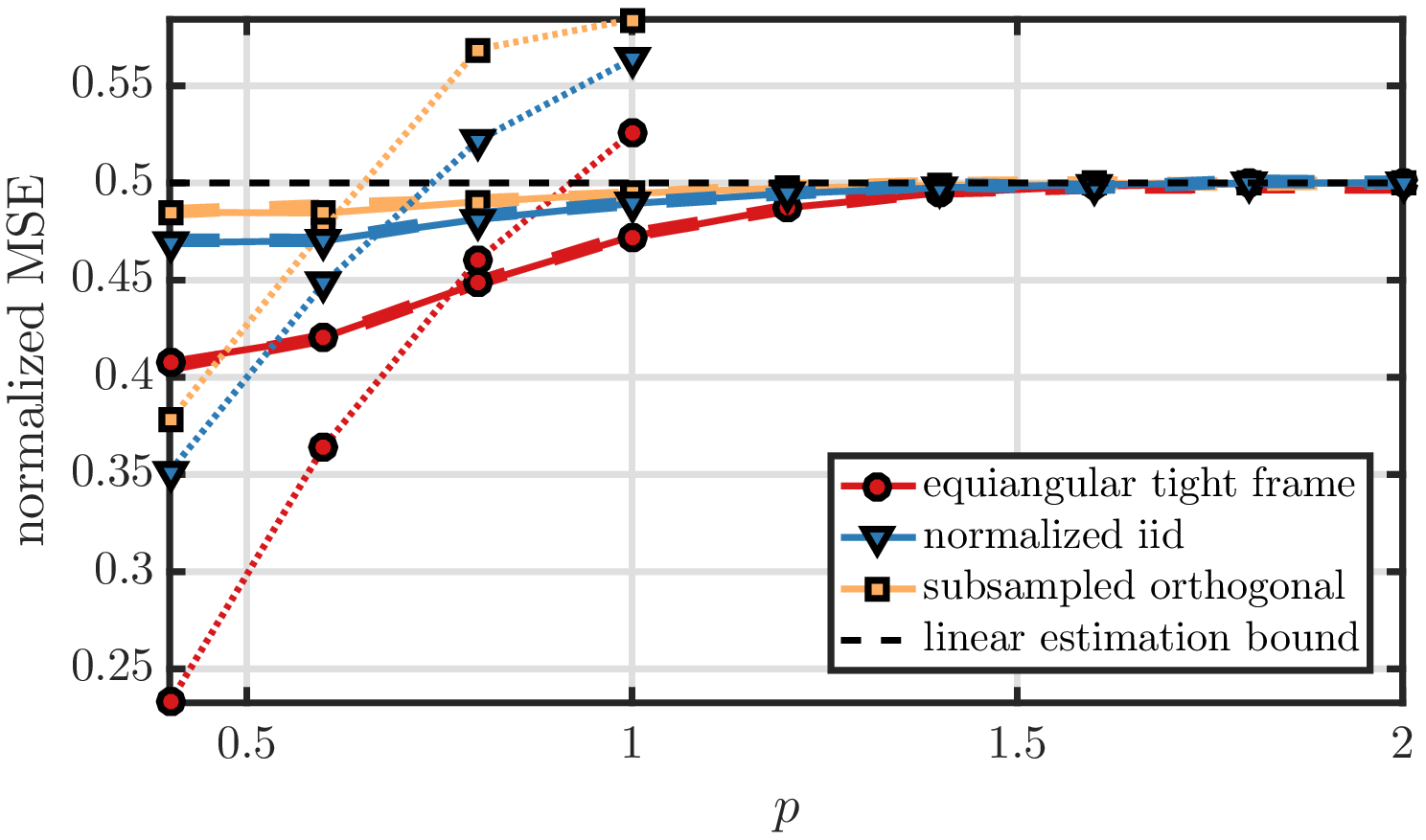}\vspace{-9pt}
 \caption{Normalized MSE for the proposed estimator and varying matrices $\bm{A}$. Analytical results from \eqref{equ:smse} are shown in solid, Monte-Carlo results in dashed and results for $\ell_1$-minimization in dotted linestyle.}
\label{fig:sim_results}
\end{figure}
\vspace{-9pt}
\begin{figure}[htb]
\centering
  \includegraphics[width= .93\linewidth]{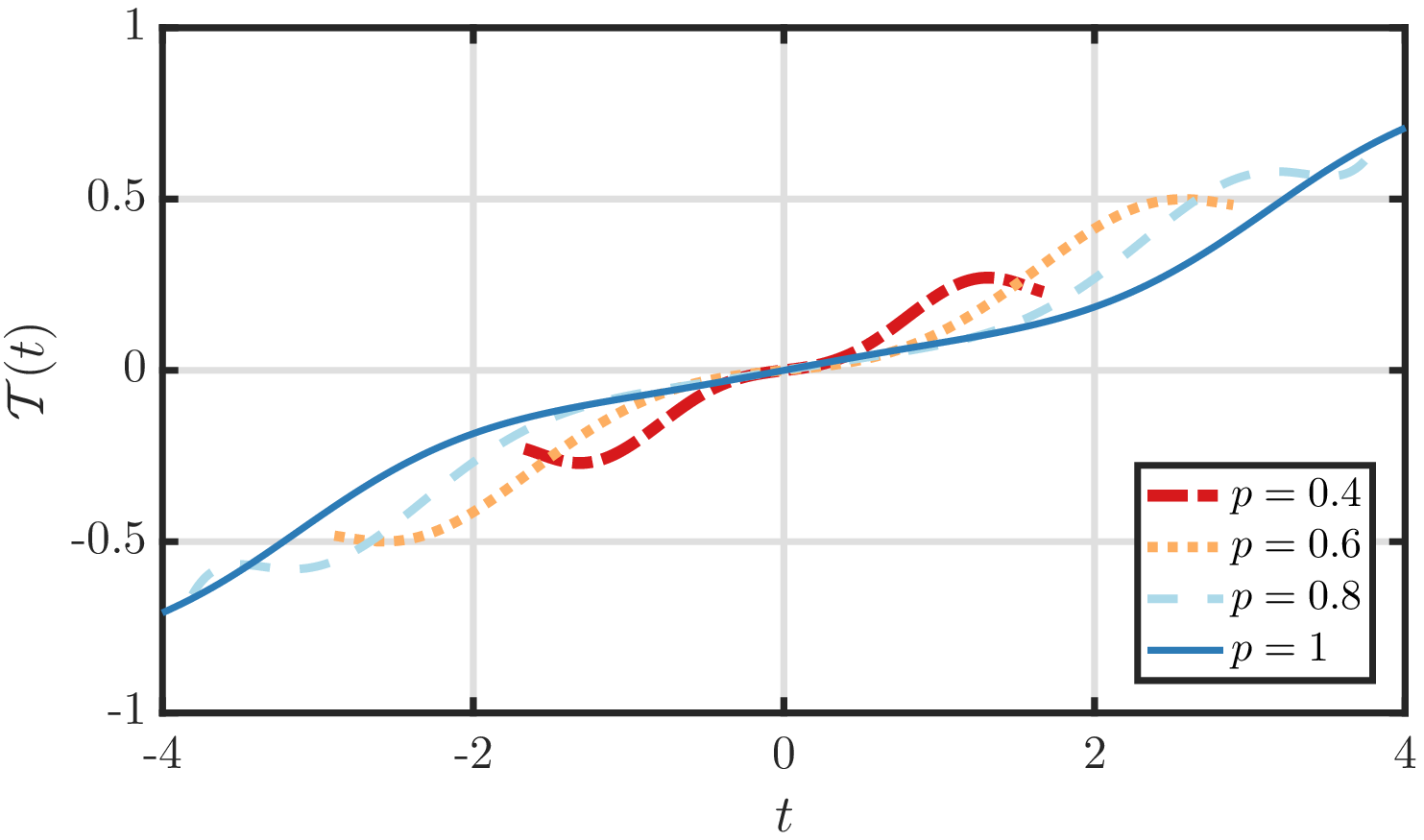}\vspace{-9pt}
 \caption{Optimal nonlinearities $\mc{T}$ for equiangular tight frame $\bm{A}_1$ and varying values of $p$ plotted over $\pm \rs{sup}_{\bm{x} \in \mc{B}_\bm{p}} \lVert \bm{W}_\diamond \bm{A}_1 \bm{x} \rVert_{\infty}$. }
\label{fig:sim_pcoeffs_etf}
\end{figure}
%\vspace{-18pt}
\newpage
\section{Conclusion}
In this paper we proposed a structured nonlinear Bayesian MMSE estimator to recover sparse signals from fixed dimensionality reducing maps. By using alternating optimization to obtain the proposed estimator composed of linear mapping and a Cartesian product of polynomial nonlinearities, we obtain a real-time capable estimator, that we show is comparable to the much more complex $\ell_1$-decoder in the low-dimensional setting. To scale to higher dimensions, a main difficulty is to obtain faster estimates of higher-order inner-product statistics. Also, using different approximation bases with faster convergence properties like trigonometric, rational or Chebyshev polynomials, may be beneficial to achieve even better estimation performance in possibly larger dimensions.

%\subsection*{Acknowledgments}

%\IEEEtriggeratref{2}
%\bibliographystyle{IEEEbib}
%\bibliography{refs}
%\newpage
%\IEEEtriggeratref{17}
%\IEEEtriggercmd{\enlargethispage{-25in}}
\bibliographystyle{IEEEbib}
\bibliography{refs}

\subsection*{Appendix}
\appendix
\section{Proof of expectation of monomials over $\mc{B}_\bm{p}$}\label{app:A}
Given the symmetry of the integration domain w.r.t. each $x_n$, it follows that the integral vanishes if at least one exponent $\alpha_n$ is odd. For the remaining part we use the fact that $\forall \bg{\alpha} \in 2\bb{N}_0^N$ the injective substitution $\varphi:\bb{R}_+^N \to \bb{R}_+^N: [x_1,\hdots,x_N] \mapsto [y_1^{1/(\alpha_1+1)},\hdots, y_N^{1/(\alpha_N+1)}]$ has Jacobian determinant 
\begin{align}
\lvert \rs{det}(J\varphi) \rvert = \prod_{n=1}^N  \frac{1}{\alpha_n + 1} \left\lvert y_n \right\rvert^{- \frac{\alpha_n}{\alpha_n + 1}}.
\end{align}
The transformed integral of \eqref{equ:int_mon} is then given by
\begin{align}
& \prod_{n=1}^N \frac{1}{\alpha_n + 1} \int_{\Omega^\prime}\prod_{n=1}^N \lvert y_n \rvert^{\frac{\alpha_n}{\alpha_n + 1}} \left\lvert y_n \right\rvert^{- \frac{\alpha_n}{\alpha_n + 1}}  \ d\bm{y} = \\
& = \prod_{n=1}^N \frac{1}{\alpha_n + 1} \int_{\Omega^\prime} 1 \ d\bm{y}, \label{equ:proof11}
\end{align}
with transformed integration domain
\begin{align}
\Omega^\prime = \sum_{n=1}^N \lvert y_n \rvert^{\frac{p_n}{\alpha_n + 1}} =: \mc{B}_{\bm{p}^\prime} \text{ with } \ \forall n: p_n^\prime = \frac{p_n}{\alpha_n+1}. \label{equ:proof12}
\end{align}
Using the volume of generalized balls from \eqref{equ:uni_lp} with the characteristic vector $\bm{p}^\prime$ from \eqref{equ:proof12} in \eqref{equ:proof11} establishes the desired result.

\section{Derivation of partial derivatives}\label{app:B}
Due to linearity we may exchange the roles of trace and expectation and employ the following results on derivatives of traces  \cite{PePe08}
\begin{align}
\frac{\partial}{\partial \bm{W}} \rs{tr} \left\{ g(\bm{W}) \right\} = g^{\prime}(\bm{W})^T \\
\frac{\partial}{\partial \bm{W}} \rs{tr} \left\{ \bm{W} \bm{A} \right\} = \bm{A}^T.
\end{align}
Thus, for $d \in \bb{N}$ we have that
\begin{align}\label{equ:appb}
& \frac{\partial}{\partial \bm{W}}  \rs{tr} \left\{ \bb{E}_\bs{x} \left[ \rs{diag}^d(\bm{W}\bm{A} \bs{x}) \rs{diag}(\bs{x}) \right] \right\}  = \\
& = \frac{\partial}{\partial \bm{W}} \bb{E}_\bs{x} \left[  \rs{tr} \left\{ \bm{I} \odot (\bm{W}\bm{A}\bs{x} \bold{1}^T)^{\odot d}  \rs{diag}(\bs{x}) \right\} \right]  \nonumber \\
& = \bb{E}_\bs{x} \left[ d \bm{I} \odot (\bm{W}\bm{A}\bs{x} \bold{1}^T)^{\odot d-1} \rs{diag}(\bs{x}) \frac{\partial}{\partial \bm{W}} \rs{tr}\left\{ \bm{W}\bm{A}\bs{x}\bold{1}^T \right\} \right]  \nonumber \\
%& = d \cdot \bb{E}_\bs{x} \left[ \rs{diag}^{d-1}(\bm{W} \bm{A} \bs{x}) \rs{diag}(\bs{x}) \bold{1}\bs{x}^T \bm{A}^T \right]   \nonumber \\
& = d \cdot \bb{E}_\bs{x} \left[ \rs{diag}^{d-1}(\bm{W} \bm{A} \bs{x}) \bs{x} \bs{x}^T \bm{A}^T \right], \nonumber
\end{align}
which proves the first part, while the second part follows along similar lines by replacing $\rs{diag}(\bs{x})$ with $\bm{I}$ in \eqref{equ:appb}.

%\end{appendix}

\end{document}